\documentclass[draft,onecolumn]{IEEEtran}
\IEEEoverridecommandlockouts
\usepackage{cite}
\usepackage{amsmath,amssymb,amsfonts,amsthm}
\usepackage{txfonts}
\usepackage{algorithmic}
\usepackage{graphicx}
\usepackage{caption}
\usepackage{textcomp}
\usepackage{xcolor}
\usepackage{tikz}
\usetikzlibrary{positioning, shapes.geometric}
\usepackage{setspace} 
\usepackage{comment}
\usepackage{url}

\usepackage{blkarray}
\newif\ifcomment
\commenttrue

\begin{document}

\newtheorem{theorem}{Theorem}
\newtheorem*{theorem*}{Theorem}
\newtheorem{proposition}{Proposition}
\newtheorem{definition}{Definition}
\newtheorem{corollary}{Corollary}
\newtheorem{lemma}{Lemma}
\newtheorem*{lemma*}{Lemma}
\newtheorem{example}{Example}
\newtheorem{remark}{Remark}

\title{Existence and Constructions of Strict Function-Correcting Codes with Data Protection}

\author{\IEEEauthorblockN{Charul Rajput\IEEEauthorrefmark{1}, B. Sundar Rajan\IEEEauthorrefmark{2}, Ragnar Freij-Hollanti\IEEEauthorrefmark{3}, Camilla Hollanti\IEEEauthorrefmark{3}}

\IEEEauthorblockA{\IEEEauthorrefmark{1}Signal Processing and Communication Research Center (SPCRC), \\
The International Institute of Information Technology Hyderabad, India \\
Email: charul.rajput@research.iiit.ac.in}
\linebreak
\IEEEauthorblockA{\IEEEauthorrefmark{2}Department of Electrical Communication Engineering,
Indian Institute of Science, \\ Bengaluru, India 
Email: bsrajan@iisc.ac.in}
\linebreak
\IEEEauthorblockA{\IEEEauthorrefmark{3}Department of Mathematics and Systems Analysis,
Aalto University, Finland \\
Emails: \{ragnar.freij, camilla.hollanti\}@aalto.fi}}



\maketitle

\vspace{-8mm}
\begin{abstract}
Function-correcting codes with data protection simultaneously protect both the data and a function of the data at distinct error-correction levels. When the function receives strictly stronger protection than the data, such a code is called a strict function-correcting code with data protection. While prior work established that well-known code families such as perfect codes and MDS codes cannot serve as strict function-correcting codes, the question of which codes can serve this role, and how to construct them, has remained open.

In this paper, we address the existence and construction of strict function-correcting codes for linear codes through three main contributions. First, using the $\alpha$-distance graph framework introduced in our prior work, we establish a graph-theoretic existence condition under which a code can serve as a strict function-correcting code. For linear codes, we prove that this distance graph is isomorphic to a Cayley graph, which implies that the connected components are cosets of the subcode generated by low-weight codewords. This transforms the existence problem into a subcode generation problem.
Second, a classical result of Simonis shows that any linear code can be transformed into one with the same parameters whose basis consists entirely of minimum-weight codewords. We develop a converse construction: under certain conditions on the weight distribution, a linear code can be transformed into a new code with the same parameters but fewer independent minimum-weight codewords, thereby producing codes suitable for use as strict function-correcting codes. As a source of codes satisfying these conditions, we introduce chain codes, an infinite family of linear codes that are generated by their minimum-weight codewords.
Third, we present an independent construction of strict function-correcting codes from narrow-sense BCH codes with designed distance three, by proving that the minimum-weight codewords of such codes are contained in a proper subcode.
\end{abstract}


\section{Introduction}
 Consider a sender who encodes a message and transmits it over a noisy channel to a receiver. In some scenarios, the receiver does not need to recover the full message but is only interested in evaluating a particular function of the message. Such situations arise naturally in machine learning applications, where a correct classification decision or prediction is often more important than exact recovery of the underlying data or parameters. Protecting the entire message via a classical error-correcting code would guarantee that the function value can be recovered, but this may require more redundancy than necessary.
If the function of interest is known to the sender, the message can be encoded to exploit the structure of the function with potentially lower redundancy. This idea was formalized by Lenz et al. \cite{LBWY2023}, who introduced \emph{function-correcting codes} (FCCs). In this framework, a systematic encoding is designed so that codewords corresponding to different function values are separated by a prescribed minimum distance, while no distance constraint is imposed between codewords sharing the same function value. The authors of \cite{LBWY2023} established the equivalence between FCCs and irregular-distance codes and used this connection to derive bounds on the optimal redundancy for various function classes.
 
A limitation of the original FCC setup is that no error protection is guaranteed for the data itself. The framework of \cite{LBWY2023} ensures reliable recovery of the function value, but the underlying message may be corrupted beyond recovery. In many applications, such as networks where different nodes compute different functions of the same data, or storage systems where certain attributes are more critical than others, it is desirable to provide at least a baseline level of error protection for the data alongside the stronger protection for the function. Motivated by this, Rajput et al. \cite{RRFH2025} introduced \emph{function-correcting codes with data protection}. In this generalized framework, the encoding satisfies two distance requirements: a minimum distance~$d_d$ between all pairs of distinct codewords ensures data protection, while a potentially larger minimum distance~$d_f \geq d_d$ between codewords with different function values ensures function protection. When the inequality is strict, i.e., $d_f > d_d$, the code is called a \emph{strict} FCC with data protection. A two-step construction procedure and bounds on the optimal redundancy were presented in~\cite{RRFH2025}, along with examples demonstrating that data protection can sometimes be added to existing FCCs without increasing redundancy.
 
In a follow-up work \cite{RRFH2026}, the authors investigated whether well-known families of classical codes can serve as strict FCCs. By associating to each code a graph based on the pairwise Hamming distances of its codewords, they showed that codes whose distance graph is sufficiently connected cannot provide strictly stronger protection for any nontrivial function. In particular, perfect codes and maximum distance separable (MDS) codes were shown to be unsuitable. These non-existence results naturally motivate the study of codes that \emph{can} serve as strict function-correcting codes with data protection. In this work, we address this problem by establishing existence conditions and providing explicit constructions of such codes.
 
\subsection{Related Work}
 
Function-correcting codes were introduced in~\cite{LBWY2023}, and since then have been studied along several directions. Premlal and Rajan~\cite{PR2025} developed a graph-based framework and derived tighter redundancy bounds, focusing on linear functions and connections to classical systematic codes. Ge, Xu, Zhang, and Zhang~\cite{GXZZ2025} obtained near-optimal redundancy bounds for Hamming weight and Hamming weight distribution functions. Ly and Soljanin~\cite{LS2025} derived redundancy bounds over general finite fields and established their tightness for sufficiently large fields.
 
Function-correcting codes have been extended to channel models beyond the binary symmetric channel. Xia et al. ~\cite{XLC2024} introduced FCCs for symbol-pair read channels, and Singh et al.~\cite{SSY2025} generalized this to $b$-symbol read channels over finite fields. Sampath and Rajan~\cite{SR2025} studied linear functions in the $b$-symbol setting and derived Plotkin-type bounds. In the Lee metric, \cite{VS2025} and \cite{HUR2025} developed theoretical foundations and explicit constructions. The work in~\cite{LL2025} extended the framework to codes with homogeneous distance.
 
In~\cite{RRFH2025a}, authors proposed FCCs for locally bounded functions and derived redundancy upper bounds. The work in~\cite{RRFH2026a} introduced function-correcting partition codes, which operate directly on a partition of the message domain rather than on a specific function, enabling a single code to protect multiple functions simultaneously with applications in broadcast networks.
 
On the constructive side for FCCs with data protection, Durgi et al.~\cite{DMKPR2026} studied FCCs for the Hamming code membership function with optimal data protection, while ~\cite{DMKPR2026a} investigated the role of distance-matrix structure in the error performance of FCCs for maximally unbalanced Boolean functions.
 
Two classical results are closely related to our constructions. Simonis~\cite{S1992} proved that any linear code can be transformed into a code with the same parameters whose generator matrix consists entirely of minimum-weight codewords. Mogilnykh and Solov'eva~\cite{MS2020} showed that for primes $p \geq 5$, neither the narrow-sense BCH code $C_{1,2}$ of length $p^m - 1$ over $\mathbb{F}_p$ nor its extension is generated by its minimum-weight codewords, and that the extended code admits a basis of weight-5 codewords. Their analysis of the subcode structure of these BCH codes is closely related to our results in Section~\ref{sec:bch}.

\subsection{Contributions}
 
The main contributions of this paper are as follows.
 
\begin{enumerate}
\item Using the $\alpha$-distance graph framework introduced in~\cite{RRFH2026}, we establish
sufficient conditions under which a code can serve as a strict FCC with data protection (Theorem~\ref{thm:existence}).
 
\item We prove that for a linear code, the $\alpha$-distance graph is a Cayley graph on the code viewed as an additive group (Lemma~\ref{lem:cayley}). As a consequence, the connected components are cosets of the subcode generated by codewords of weight at
most~$\alpha$, and the existence problem reduces to a subcode generation problem.
 
\item Simonis~\cite{S1992} showed that any linear code can be transformed into one with the
same parameters whose basis consists entirely of minimum-weight codewords. We develop a
converse: under conditions on the weight distribution, a code can be transformed into a new
code with the same parameters but fewer independent minimum-weight codewords
(Theorems~\ref{thm:reverse1} and~\ref{thm:reverse2}).
 
\item We introduce \emph{chain codes}, a family of linear codes that are generated by their minimum-weight codewords and satisfy the conditions required by the converse construction (Propositions~\ref{prop:chain1} and~\ref{prop:chain2}).
 
\item We prove that for primes $p \geq 5$ and odd $m \geq 3$, the weight-3 codewords of the
narrow-sense BCH code $C_{1,2}$ over $\mathbb{F}_p$ are contained in the subcode
$C_{1,2,\, p+1,\, p^2+1,\, \ldots,\, p^{(m-1)/2}+1}$ (Theorem~\ref{thm:bch}). This
yields explicit strict $(f\!:\!3,4)$-FCC constructions.
\end{enumerate}
 
\subsection{Organization}
 
The paper is organized as follows. Section~\ref{sec:prelim} recalls the necessary background on FCCs with data protection. Section~\ref{sec:existence} establishes the graph-theoretic existence condition for strict FCCs and develops the Cayley graph characterization for linear codes. Section~\ref{sec:reverse} presents the converse of Simonis's theorem and the associated construction. Section~\ref{sec:chain} introduces chain codes and establishes their properties. Section~\ref{sec:bch} develops the FCC construction from BCH codes. Section~\ref{sec:conclusion} concludes the paper.

\subsection{Notations}

$\mathbb{F}_q$ denote the finite field with $q$ elements, where $q$ is a prime power. For a positive integer $n$, the set $\{1, 2, \ldots, n\}$ is denoted by $[n]$. For two vectors $x, y \in \mathbb{F}_q^n$, the Hamming distance $d(x, y)$ is the number of coordinates in which $x$ and $y$ differ. The Hamming weight of a vector $x \in \mathbb{F}_q^n$ is $\mathrm{wt}(x) = d(x, 0)$, and its support is $\mathrm{supp}(x) = \{i \in [n] : x_i \neq 0\}$.
A block code $C$ over $\mathbb{F}_q$ of length $n$, size $M = |C|$, and minimum distance $d = \min\{d(x,y) : x, y \in C,\, x \neq y\}$ is called an $(n, M, d)_q$ code. A linear code $C$ over $\mathbb{F}_q$ of length $n$, dimension $k$, and minimum distance $d$ is denoted by $[n, k, d]_q$; in this case $M = q^k$. The number of codewords of weight $w$ in $C$ is denoted by $A_w(C)$. For a subcode $D \subseteq C$, the codimension of $D$ in $C$ is defined as $\mathrm{codim}(D) = \dim(C) - \dim(D)$.


\section{Preliminaries}
\label{sec:prelim}

\subsection{Function-Correcting Codes with Data Protection}

Function-correcting codes were introduced in~\cite{LBWY2023} to protect a function value against errors without necessarily protecting the data. The following generalization, introduced in~\cite{RRFH2025}, assigns separate levels of error protection to the data and the function value.

\begin{definition}[FCC with Data Protection {\cite{RRFH2025}}]
\label{def:fcc-dp}
Consider a function $f \colon \mathbb{F}_q^k \to \mathrm{Im}(f)$. An encoding $\mathcal{C}_f \colon \mathbb{F}_q^k \to \mathbb{F}_q^{k+r}$ is called an $(f : d_d, d_f)$-FCC if:
\begin{itemize}
\item for any $u_1, u_2 \in \mathbb{F}_q^k$ with $u_1 \neq u_2$,
\[
d(\mathcal{C}_f(u_1), \mathcal{C}_f(u_2)) \geq d_d,
\]
\item for any $u_1, u_2 \in \mathbb{F}_q^k$ with $f(u_1) \neq f(u_2)$,
\[
d(\mathcal{C}_f(u_1), \mathcal{C}_f(u_2)) \geq d_f,
\]
\end{itemize}
where $d_d$ and $d_f$ are non-negative integers with $d_d \leq d_f$.
\end{definition}

The parameter $d_d$ ensures that distinct messages are mapped to distinct codewords at minimum distance $d_d$, providing error protection for the data. The parameter $d_f$ provides the protection for the function value.

\begin{definition}[Strict FCC {\cite{RRFH2026}}]
\label{def:strict}
An $(f : d_d, d_f)$-FCC is called \emph{strict} if $d_f > d_d$.
\end{definition}

\begin{remark}
\label{rem:systematic}
In the original FCC framework of~\cite{LBWY2023}, only systematic codes are considered. This is because if non-systematic codes are allowed, multiple messages with the same function value could be mapped to the same codeword, making them indistinguishable at the decoder. In contrast, for FCCs with data protection, the encoding need not be systematic: the condition $d_d > 0$ already ensures that distinct messages are mapped to distinct codewords.
\end{remark}

\subsection{Distance Graphs and Non-Existence Results}

The following graph-theoretic framework was introduced in~\cite{RRFH2026} to study the feasibility of strict FCCs.

\begin{definition}[$\alpha$-Distance Graph {\cite{RRFH2026}}]
\label{def:alpha-graph}
Let $C$ be a code with minimum distance $d_{\min}(C)$. For $\alpha \geq d_{\min}(C)$, the $\alpha$-distance graph of $C$, denoted by $G_\alpha(C)$, is the graph whose vertex set is $C$ and in which two distinct vertices $c_1, c_2 \in C$ are adjacent if and only if $d(c_1, c_2) \leq \alpha$.
\end{definition}

The graph $G_\alpha(C)$ is a subgraph of $G_{\alpha'}(C)$ whenever $\alpha \leq \alpha'$. In particular, if $G_\alpha(C)$ is connected, then $G_{\alpha'}(C)$ is also connected for all $\alpha' \geq \alpha$. Moreover, $G_{d_{\max}}(C)$ is the complete graph, where $d_{\max}=\max\{d(x,y) : x, y \in C,\, x \neq y\}$.

\begin{example}
\label{ex:aDG}
Let $C = \{00000,\, 00001,\, 00010,\, 01111,\, 10111,\, 11111\} \subset \mathbb{F}_2^5$. The $\alpha$-distance graph with $\alpha = 2 > d_{\min}(C)$, i.e., $G_2(C)$, has edges only between pairs at distance at most~$2$. The graph consists of two disjoint triangles, as shown in Fig.~\ref{fig:aDG}.

\begin{figure}[t]
\centering
\begin{tikzpicture}[scale=0.8,
  every node/.style={circle,draw,inner sep=2pt,minimum size=18pt}]
  \node (a) at (0,0)   {\small 00000};
  \node (b) at (-1.8,1.6) {\small 00001};
  \node (c) at (1.8,1.6)  {\small 00010};

  \node (d) at (5.2,0)   {\small 11111};
  \node (e) at (3.4,1.6) {\small 01111};
  \node (f) at (7,1.6) {\small 10111};

  \draw (a)--(b);
  \draw (a)--(c);
  \draw (b)--(c);

  \draw (d)--(e);
  \draw (d)--(f);
  \draw (e)--(f);
\end{tikzpicture}
\caption{The $\alpha$-distance graph $G_2(C)$ for Example~\ref{ex:aDG}.}
\label{fig:aDG}
\end{figure}
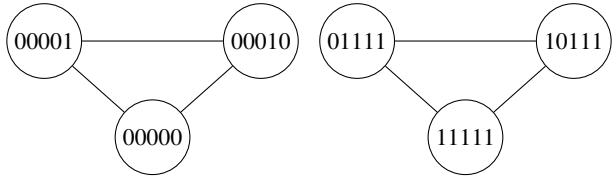
\end{example}

The non-existence results established in~\cite{RRFH2026} demonstrate that the structure of the distance graph imposes strong constraints on the existence of strict FCCs.

\begin{theorem}[{\cite{RRFH2026}}]
\label{thm:nonexist1}
Let $C$ be an $(n, q^k, d)$ code. If the $\alpha$-distance graph $G_\alpha(C)$ is connected, then $C$ cannot be an $(f : d, d_f)$-FCC for any $f \colon \mathbb{F}_q^k \to \mathrm{Im}(f)$ with $|\mathrm{Im}(f)| \geq 2$ and $d_f > \alpha$.
\end{theorem}

\begin{theorem}[{\cite{RRFH2026}}]
\label{thm:nonexist2}
Let $C$ be an $(n, q^k, d)$ code. If the $\alpha$-distance graph $G_\alpha(C)$ has $Q$ connected components, then $C$ cannot be an $(f : d, d_f)$-FCC for any $f \colon \mathbb{F}_q^k \to \mathrm{Im}(f)$ with $|\mathrm{Im}(f)| \geq Q + 1$ and $d_f > \alpha$.
\end{theorem}

\subsection{Cayley Graphs}

We recall the definition of a Cayley graph, which will be used in Section~\ref{sec:cayley} to characterize the distance graphs of linear codes.

\begin{definition}[Cayley Graph]
\label{def:cayley}
Let $G$ be a group with identity element $e$, and let $S \subseteq G \setminus \{e\}$ be a subset that is closed under taking inverses. The Cayley graph of $G$ with respect to $S$, denoted by $\mathrm{Cay}(G, S)$, is the undirected graph with vertex set $G$ in which two vertices $g, h \in G$ are adjacent if and only if $g - h \in S$.
\end{definition}

The following property of Cayley graphs is well known (see, e.g.,~\cite{KS2007}).

\begin{proposition}
\label{prop:cayley-components}
Let $\mathrm{Cay}(G, S)$ be a Cayley graph. If the set $S$ does not generate $G$, then $\mathrm{Cay}(G, S)$ is disconnected, and each connected component is a coset of the subgroup $\langle S \rangle$ generated by $S$. In particular, all connected components have the same size $|\langle S \rangle|$, and the number of connected components is $|G|/|\langle S \rangle|$.
\end{proposition}

\begin{example}
\label{ex:cayley-Z4}
Consider the group $G = \mathbb{Z}_4 = \{0, 1, 2, 3\}$ under addition modulo~$4$, with identity element $e = 0$. Let $S = \{1, 3\} \subseteq \mathbb{Z}_4 \setminus \{0\}$. Since $-1 \equiv 3 \pmod{4}$ and $-3 \equiv 1 \pmod{4}$, the set $S$ is closed under taking inverses, and $\mathrm{Cay}(\mathbb{Z}_4, S)$ is undirected. The edges are $\{0,1\}$, $\{1,2\}$, $\{2,3\}$, and $\{3,0\}$, so $\mathrm{Cay}(\mathbb{Z}_4, S)$ is the cycle graph $C_4$, as shown in Fig.~\ref{fig:cayley-Z4}. Since $S$ generates $\mathbb{Z}_4$, the graph is connected.

\begin{figure}[t]
\centering
\begin{tikzpicture}[scale=1.2,
  every node/.style={circle, draw, inner sep=2.2pt, font=\small},
  edge/.style={thick}
]
  \node (v0) at (0,0) {0};
  \node (v1) at (2,0) {1};
  \node (v2) at (2,2) {2};
  \node (v3) at (0,2) {3};

  \draw[edge] (v0) -- (v1);
  \draw[edge] (v1) -- (v2);
  \draw[edge] (v2) -- (v3);
  \draw[edge] (v3) -- (v0);

  \node[draw=none, rectangle, font=\normalsize] at (1,-0.8)
    {$\mathrm{Cay}(\mathbb{Z}_4,\{1,3\}) \cong C_4$};
\end{tikzpicture}
\caption{The Cayley graph $\mathrm{Cay}(\mathbb{Z}_4, \{1,3\})$ for Example~\ref{ex:cayley-Z4}.}
\label{fig:cayley-Z4}
\end{figure}
\end{example}

\subsection{BCH Codes}

We recall the definition of narrow-sense BCH codes, which will be used in Section~\ref{sec:bch}.

Let $\alpha$ be a primitive element of $\mathbb{F}_{p^m}$, where $p$ is a prime and $m \geq 3$. For an integer $i$, the $p$-cyclotomic coset of $i$ modulo $p^m - 1$ is
\[
\mathrm{Cl}(i) = \{i, ip, ip^2, \ldots, ip^{\ell_i - 1}\} \pmod{p^m - 1},
\]
where $\ell_i = |\mathrm{Cl}(i)|$ divides $m$.

\begin{definition}[Narrow-Sense BCH Code]
\label{def:bch}
The narrow-sense BCH code $C_{1,2,\ldots,\delta-1}$ of length $n = p^m - 1$ over $\mathbb{F}_p$ with designed distance $\delta$ is the cyclic code with defining set
\[
T_{1,2,\ldots,\delta-1} = \mathrm{Cl}(1) \cup \mathrm{Cl}(2) \cup \cdots \cup \mathrm{Cl}(\delta - 1).
\]
The BCH bound guarantees that the minimum distance satisfies $d \geq \delta$, and the dimension is
\[
\dim(C_{1,2,\ldots,\delta-1}) = n - \left| \bigcup_{i=1}^{\delta-1} \mathrm{Cl}(i) \right|.
\]
\end{definition}

In particular, the code $C_{1,2}$ has designed distance $\delta = 3$ and dimension $\dim(C_{1,2}) = p^m - 1 - 2m$ (when $\mathrm{Cl}(1)$ and $\mathrm{Cl}(2)$ are disjoint and each of size $m$). The following result on the minimum distance of $C_{1,2}$ was established in~\cite{CTZ1999}.

\begin{proposition}[{\cite{CTZ1999}}]
\label{prop:bch-distance}
The minimum distance of the BCH code $C_{1,2}$ is $3$ for all primes $p \neq 3$, and $4$ for $p = 3$. The extensions of these codes have minimum distances $4$ and $5$, respectively.
\end{proposition}

\section{Existence Conditions for Strict Function-Correcting Codes}
\label{sec:existence}

The non-existence results of Theorems~\ref{thm:nonexist1} and~\ref{thm:nonexist2} show that codes with connected distance graphs cannot serve as strict FCCs. We now establish the converse direction: sufficient conditions under which a code \emph{can} serve as a strict FCC with data protection.

\begin{theorem}
\label{thm:existence}
Let $f \colon \mathbb{F}_q^k \to S$ and denote $E = |\mathrm{Im}(f)|$. Let $C \subseteq \mathbb{F}_q^n$ be a code satisfying the following:
\begin{enumerate}
\item[\textup{(C1)}] $d_{\min}(C) = d_d$.
\item[\textup{(C2)}] There exist $E$ pairwise disjoint unions of connected components $\mathcal{C}_1, \ldots, \mathcal{C}_E$ of $G_{d_f - 1}(C)$ such that
\[
|\mathcal{C}_i| = |f^{-1}(a_i)| \quad \text{for some ordering } a_1, \ldots, a_E \in \mathrm{Im}(f).
\]
\end{enumerate}
Then $C$ is an $(f : d_d, d_f)$-FCC.
\end{theorem}

\begin{proof}
Using the pairwise disjoint unions of connected components $\mathcal{C}_1, \ldots, \mathcal{C}_E \subseteq C$ of $G_{d_f-1}(C)$, we define an encoding $\mathcal{C} \colon \mathbb{F}_q^k \to C$ as follows. Let $a_1, a_2, \ldots, a_E$ be the ordering of $\mathrm{Im}(f)$ considered in~(C2). For each $i \in [E]$, let $\varphi_i \colon f^{-1}(a_i) \to \mathcal{C}_i$ be a bijection. Define
\[
\mathcal{C}(u) = \varphi_i(u), \quad \text{if } f(u) = a_i, \; i \in [E].
\]

Since $d_{\min}(C) = d_d$, the first condition of Definition~\ref{def:fcc-dp} is satisfied.

Now let $u, v \in \mathbb{F}_q^k$ with $f(u) = a_i$ and $f(v) = a_j$, where $i \neq j$. Then $\mathcal{C}(u) = \varphi_i(u) \in \mathcal{C}_i$ and $\mathcal{C}(v) = \varphi_j(v) \in \mathcal{C}_j$. Since $\mathcal{C}_i$ and $\mathcal{C}_j$ are pairwise disjoint unions of connected components of $G_{d_f - 1}(C)$, there is no edge between any vertex in $\mathcal{C}_i$ and any vertex in $\mathcal{C}_j$. Therefore $d(x, y) \geq d_f$ for all $x \in \mathcal{C}_i$ and $y \in \mathcal{C}_j$, and in particular,
\[
d(\mathcal{C}(u), \mathcal{C}(v)) = d(\varphi_i(u), \varphi_j(v)) \geq d_f.
\]
Hence $C$ is an $(f : d_d, d_f)$-FCC.
\end{proof}

Condition~(C2) can equivalently be stated in terms of the function domain partition $\mathcal{P}_f = \{P_1, \ldots, P_E\}$ of $\mathbb{F}_q^k$, where $P_i = f^{-1}(a_i)$. The requirement is that the connected components of $G_{d_f - 1}(C)$ can be grouped into $E$ disjoint collections whose sizes match the block sizes of $\mathcal{P}_f$.

\begin{example}
\label{ex:existence}
Consider the code $C = \{000000000,\allowbreak\, 100110110,\allowbreak\, 010101110,\allowbreak\, 001011110,\allowbreak\, 110011101,\allowbreak\, 101101101,\allowbreak\, 011110101,\allowbreak\, 111000011\} \subseteq \mathbb{F}_2^9$ with minimum distance $d_d = 4$. The $\alpha$-distance graph $G_4(C)$ has four connected components:
\begin{align*}
\mathcal{S}_1 &= \{000000000\}, &
\mathcal{S}_2 &= \{100110110,\, 010101110,\, 001011110\}, \\
\mathcal{S}_3 &= \{110011101,\, 101101101,\, 011110101\}, &
\mathcal{S}_4 &= \{111000011\},
\end{align*}
as shown in Fig.~\ref{fig:Graph}.

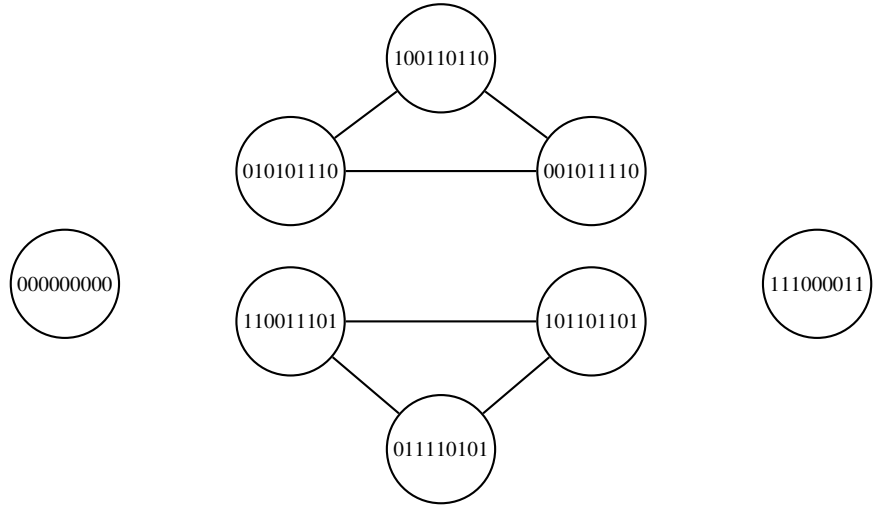
\begin{figure}[t]
\centering
\begin{tikzpicture}[
  vertex/.style={
    circle,
    draw,
    thick,
    minimum size=9mm,
    inner sep=2pt,
    font=\footnotesize
  }
]

\node[vertex] (a) at (0,3) {100110110};
\node[vertex] (b) at (-2,1.5) {010101110};
\node[vertex] (c) at (2,1.5) {001011110};

\draw[thick] (a)--(b)--(c)--(a);

\node[vertex] (d) at (-2,-0.5) {110011101};
\node[vertex] (e) at (2,-0.5) {101101101};
\node[vertex] (f) at (0,-2.2) {011110101};

\draw[thick] (d)--(e)--(f)--(d);

\node[vertex] (g) at (-5,0) {000000000};
\node[vertex] (h) at (5,0) {111000011};

\end{tikzpicture}
\caption{The distance graph $G_4(C)$ for Example~\ref{ex:existence}.}
\label{fig:Graph}
\end{figure}

Now consider the function $f \colon \mathbb{F}_2^3 \to \mathbb{F}_2^2$ defined by
\[
f(x_1, x_2, x_3) = (x_1 x_2,\; x_3(1 + x_1 x_2)).
\]
The function has three distinct values in its image: $00$, $01$, and $10$. The preimage sizes are $|f^{-1}(00)| = 3$, $|f^{-1}(01)| = 3$, and $|f^{-1}(10)| = 2$. We verify conditions~\textup{(C1)} and~\textup{(C2)}: the minimum distance is $d_{\min}(C) = 4 = d_d$, and the components can be grouped as $\mathcal{C}_1 = \mathcal{S}_2$ with $|\mathcal{C}_1| = 3 = |f^{-1}(00)|$, $\mathcal{C}_2 = \mathcal{S}_3$ with $|\mathcal{C}_2| = 3 = |f^{-1}(01)|$, and $\mathcal{C}_3 = \mathcal{S}_1 \cup \mathcal{S}_4$ with $|\mathcal{C}_3| = 2 = |f^{-1}(10)|$.
Therefore, $C$ is an $(f : 4, 5)$-FCC. An explicit assignment of messages to codewords is given in Table~\ref{tab:example}.

\begin{table}[t]
\centering
\renewcommand{\arraystretch}{1.15}
\caption{Codeword assignment for Example~\ref{ex:existence}.}
\label{tab:example}
\begin{tabular}{ccc}
\hline
$x \in \mathbb{F}_2^3$ & $f(x) \in \mathbb{F}_2^2$ & Assigned codeword \\
\hline
$000$ & $00$ & $100110110$ \\
$100$ & $00$ & $010101110$ \\
$010$ & $00$ & $001011110$ \\
$001$ & $01$ & $110011101$ \\
$011$ & $01$ & $101101101$ \\
$101$ & $01$ & $011110101$ \\
$110$ & $10$ & $000000000$ \\
$111$ & $10$ & $111000011$ \\
\hline
\end{tabular}
\end{table}
\end{example}

\subsection{Cayley Graph Characterization for Linear Codes}
\label{sec:cayley}

For linear codes, the $\alpha$-distance graph admits a clean algebraic characterization that simplifies the existence conditions of Theorem~\ref{thm:existence}.

\begin{lemma}
\label{lem:cayley}
Let $C \subseteq \mathbb{F}_q^n$ be a linear code. For $\alpha \geq 0$, define
\[
S_\alpha = \{c \in C : 0 < \mathrm{wt}(c) \leq \alpha\}.
\]
Then $G_\alpha(C) \cong \mathrm{Cay}(C, S_\alpha)$.
\end{lemma}

\begin{proof}
Since $C$ is linear, for any $x, y \in C$ we have $d(x, y) = \mathrm{wt}(x - y)$. Two distinct codewords $x, y \in C$ are adjacent in $G_\alpha(C)$ if and only if $d(x, y) \leq \alpha$, which holds if and only if $\mathrm{wt}(x - y) \leq \alpha$, i.e., $x - y \in S_\alpha$. This is equivalent to $y = x + s$ for some $s \in S_\alpha$, which is precisely the adjacency condition in $\mathrm{Cay}(C, S_\alpha)$. Since both graphs have the same vertex set and the same edge set, they are identical.
\end{proof}

By Proposition~\ref{prop:cayley-components}, if $S_\alpha$ does not generate $C$, then the connected components of $G_\alpha(C)$ are exactly the cosets of the subcode $\langle S_\alpha \rangle$ in $C$. This yields the following.

\begin{corollary}
\label{cor:linear-components}
Let $C$ be an $[n, k, d]_q$ linear code, and let $\alpha \geq d$. The connected components of $G_\alpha(C)$ are the cosets of the subcode $\langle S_\alpha \rangle$ in $C$. In particular, all connected components have the same size $\gamma = |\langle S_\alpha \rangle|$, and the number of connected components is $q^k / \gamma$. Moreover, $G_\alpha(C)$ is disconnected if and only if $\langle S_\alpha \rangle \subsetneq C$.
\end{corollary}

We illustrate this characterization with two examples.

\begin{example}
\label{ex:cayley-binary}
Let $C \subseteq \mathbb{F}_2^6$ be the linear code with generator matrix
\[
G = \begin{pmatrix}
1 & 1 & 0 & 0 & 0 & 0 \\
0 & 0 & 1 & 1 & 0 & 0 \\
0 & 0 & 0 & 1 & 1 & 1
\end{pmatrix}.
\]
Then $|C| = 8$ and $d_{\min}(C) = 2$. For $\alpha = 2$, we have $S_2 = \{110000,\, 001100\}$ and $\langle S_2 \rangle = \{000000,\, 110000,\, 001100,\, 111100\}$, which is a subcode of dimension~$2$. The two cosets of $\langle S_2 \rangle$ in $C$ are
\begin{align*}
\langle S_2 \rangle &= \{000000,\, 110000,\, 001100,\, 111100\}, \\
000111 + \langle S_2 \rangle &= \{000111,\, 110111,\, 001011,\, 111011\},
\end{align*}
and $G_2(C)$ consists of two connected components, each forming a 4-cycle, as shown in Fig.~\ref{fig:cayley-binary}.

\begin{figure}[t]
\centering
\begin{tikzpicture}[scale=1.05,
  every node/.style={circle, draw, inner sep=2pt, font=\small},
  edge/.style={thick}
]
  \node (a) at (0,0) {$000000$};
  \node (b) at (2,0) {$110000$};
  \node (c) at (2,2) {$111100$};
  \node (d) at (0,2) {$001100$};
  \draw[edge] (a)--(b)--(c)--(d)--(a);

  \node (e) at (5,0) {$000111$};
  \node (f) at (7,0) {$110111$};
  \node (g) at (7,2) {$111011$};
  \node (h) at (5,2) {$001011$};
  \draw[edge] (e)--(f)--(g)--(h)--(e);

  \node[draw=none, rectangle, font=\normalsize] at (1,-0.9) {$\langle S_2\rangle$};
  \node[draw=none, rectangle, font=\normalsize] at (6,-0.9) {$000111+\langle S_2\rangle$};
\end{tikzpicture}
\caption{The distance graph $G_2(C)$ for Example~\ref{ex:cayley-binary}.}
\label{fig:cayley-binary}
\end{figure}

The code $C$ can serve as a strict $(f : 2, 3)$-FCC for any function $f \colon \mathbb{F}_2^3 \to \{0, 1\}$ whose preimage sets both have size~$4$.
\end{example}

\begin{example}
\label{ex:cayley-ternary}
Let $q = 3$, $n = 4$, and let $C \subseteq \mathbb{F}_3^4$ be the linear code generated by $g_1 = (1,1,0,0)$ and $g_2 = (0,1,1,1)$. Then
\[
C = \langle g_1, g_2 \rangle = \{a g_1 + b g_2 : a, b \in \mathbb{F}_3\} = \{(a,\, a+b,\, b,\, b) : a, b \in \mathbb{F}_3\},
\]
so $|C| = 9$ and $d_{\min}(C) = 2$. For $\alpha = 2$, we have $S_2 = \{(1,1,0,0),\, (2,2,0,0)\}$ and $\langle S_2 \rangle = \{(0,0,0,0),\, (1,1,0,0),\, (2,2,0,0)\}$. The three cosets of $\langle S_2 \rangle$ in $C$ are
\begin{align*}
\mathcal{C}_0 &= \{0000,\, 1100,\, 2200\}, \\
\mathcal{C}_1 &= \{0111,\, 1211,\, 2011\}, \\
\mathcal{C}_2 &= \{0222,\, 1022,\, 2122\},
\end{align*}
and $G_2(C)$ has three connected components, each forming a triangle, as shown in Fig.~\ref{fig:cayley-ternary}.

\begin{figure}[t]
\centering
\begin{tikzpicture}[scale=1.05,
  every node/.style={circle, draw, inner sep=2pt, font=\small},
  edge/.style={thick}
]
  \node (a0) at (0,0) {$0000$};
  \node (a1) at (1.8,0) {$1100$};
  \node (a2) at (0.9,1.6) {$2200$};
  \draw[edge] (a0)--(a1)--(a2)--(a0);
  \node[draw=none, rectangle, font=\normalsize] at (0.9,-0.8) {$\mathcal{C}_0$};
  \node (b0) at (4.2,0) {$0111$};
  \node (b1) at (6.0,0) {$1211$};
  \node (b2) at (5.1,1.6) {$2011$};
  \draw[edge] (b0)--(b1)--(b2)--(b0);
  \node[draw=none, rectangle, font=\normalsize] at (5.1,-0.8) {$\mathcal{C}_1$};
  \node (c0) at (8.4,0) {$0222$};
  \node (c1) at (10.2,0) {$1022$};
  \node (c2) at (9.3,1.6) {$2122$};
  \draw[edge] (c0)--(c1)--(c2)--(c0);
  \node[draw=none, rectangle, font=\normalsize] at (9.3,-0.8) {$\mathcal{C}_2$};
\end{tikzpicture}
\caption{The distance graph $G_2(C)$ for Example~\ref{ex:cayley-ternary}.}
\label{fig:cayley-ternary}
\end{figure}
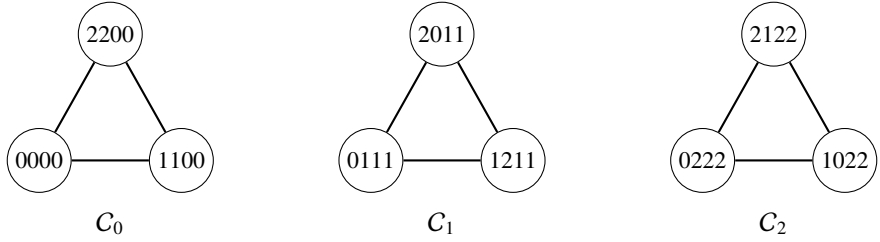

The code $C$ can serve as a strict $(f : 2, 3)$-FCC for any function $f \colon \mathbb{F}_3^2 \to S$ with $|S| \leq 3$ and each preimage of size~$3$.
\end{example}

For strict $(f : d_d, d_f)$-FCCs with $d_d = d_{\min}(C)$, the relevant graph is $G_{d_f - 1}(C)$. Applying Corollary~\ref{cor:linear-components} with $\alpha = d_f - 1$, the conditions of Theorem~\ref{thm:existence} simplify as follows.

\begin{corollary}
\label{cor:linear-existence}
Let $C$ be an $[n, k, d]_q$ linear code, let $f \colon \mathbb{F}_q^k \to S$ with $E = |\mathrm{Im}(f)|$, and let $d_f > d = d_d$. Define $\gamma = |\langle S_{d_f - 1} \rangle|$. Then $C$ can be used as a strict $(f : d, d_f)$-FCC if:
\begin{enumerate}
\item[\textup{(L1)}] $q^k / \gamma \geq E$.
\item[\textup{(L2)}] Each $|f^{-1}(a)|$ is a multiple of $\gamma$.
\end{enumerate}
\end{corollary}

Corollary~\ref{cor:linear-existence} transforms the existence problem for strict FCCs from linear codes into a \emph{subcode generation problem}: one needs to find a linear code $C$ with $d_{\min}(C) = d_d$ such that the subcode $\langle S_{d_f - 1} \rangle$ generated by codewords of weight at most $d_f - 1$ is a proper subcode of $C$, with an index that is compatible with the preimage structure of the desired function.

Note that if $C$ is generated by its minimum-weight codewords, i.e., $\langle S_d \rangle = C$, then $\langle S_\alpha \rangle = C$ for all $\alpha \geq d$, and $G_\alpha(C)$ is connected for all $\alpha \geq d$. By Theorem~\ref{thm:nonexist1}, such a code cannot serve as a strict $(f : d, d_f)$-FCC for any nontrivial function. Therefore, a necessary condition for a linear code to be usable as a strict FCC is that it is \emph{not} generated by its minimum-weight codewords. The subsequent sections develop methods for constructing codes that satisfy this condition.

\section{Converse of Simonis's Theorem}
\label{sec:reverse}

Simonis~\cite{S1992} proved that any linear code of dimension $k$ and minimum weight $d$ can be transformed into a code with the same parameters whose generator matrix consists entirely of weight-$d$ codewords. In other words, every linear code can be transformed into one that is generated by its minimum-weight codewords. As observed in Section~\ref{sec:cayley}, such codes cannot serve as strict FCCs since their distance graphs are connected. In this section, we develop a converse construction: under certain conditions on the weight distribution, a linear code can be transformed into a new code with the same parameters but fewer independent minimum-weight codewords. The resulting codes have disconnected distance graphs and are therefore suitable for use as strict FCCs.

Throughout this section, let $C$ be an $[n, k, d]_q$ code with $t > 1$ independent minimum-weight codewords $a_1, \ldots, a_t$, and let $\{a_1, \ldots, a_t, e_1, \ldots, e_{k-t}\}$ be a basis of $C$ where $\mathrm{wt}(a_i) = d$ for all $i \in [t]$ and $\mathrm{wt}(e_j) > d$ for all $j \in [k-t]$.

\subsection{One-Position Insertion}
\label{sec:reverse-one}

\begin{theorem}
\label{thm:reverse1}
Let $C$ be an $[n, k, d]_q$ code with $t > 1$ independent minimum-weight codewords and $k,d \geq 2$. Suppose:
\begin{enumerate}
\item[\textup{(i)}] $A_d(C) = t(q - 1)$.
\item[\textup{(ii)}] $A_{d+1}(C) = 0$.
\end{enumerate}
Then there exists an $[n, k, d]_q$ code $D$ with exactly $t - 1$ independent minimum-weight codewords.
\end{theorem}

\begin{proof}
Condition~(i) asserts that every weight-$d$ codeword of $C$ is a scalar multiple of one of the $t$ independent weight-$d$ basis vectors $a_1, \ldots, a_t$.

 Choose a position $p \notin \mathrm{supp}(a_t)$, which exists since $\mathrm{wt}(a_t) = d < n$ for any code with $k \geq 2$, and a nonzero scalar $\alpha \in \mathbb{F}_q^*$. Define $b = a_t + \alpha \mathbf{e}_p$, where $\mathbf{e}_p$ is the standard basis vector with $1$ at position~$p$. Then $\mathrm{wt}(b) = d + 1$. Define the code
\[
D = \langle a_1, \ldots, a_{t-1}, b, e_1, \ldots, e_{k-t} \rangle.
\]

Every codeword $v \in D$ can be written as $v = \sum_{i=1}^{t-1} \lambda_i a_i + \mu b + \sum_{j=1}^{k-t} \nu_j e_j$, with a corresponding codeword $w = \sum_{i=1}^{t-1} \lambda_i a_i + \mu a_t + \sum_{j=1}^{k-t} \nu_j e_j \in C$. Since $b$ and $a_t$ differ only at position~$p$, we have $v_p = w_p + \mu \alpha$ and $v_i = w_i$ for all $i \neq p$.

\textbf{Claim 1:} $\dim(D) = k$. If $b$ were linearly dependent on $\{a_1, \ldots, a_{t-1}, e_1, \ldots, e_{k-t}\}$, then $b \in C$, and $b - a_t = \alpha \mathbf{e}_p \in C$ would be a codeword of weight~$1 < d$, a contradiction.

\textbf{Claim 2:} $d_{\min}(D) = d$. Consider a nonzero codeword $v \in D$ with corresponding $w \in C$.
\begin{itemize}
\item \emph{Case A} ($\mu = 0$): $v = w \in C$, so $\mathrm{wt}(v) \geq d$.
\item \emph{Case B} ($\mu \neq 0$, all $\lambda_i = 0$, all $\nu_j = 0$): $v = \mu b$, so $\mathrm{wt}(v) = d + 1$.
\item \emph{Case C} ($\mu \neq 0$, at least one $\lambda_i$ or $\nu_j$ nonzero): The codeword $w \in C$ is nonzero. Since $\mu \neq 0$ and at least one other coefficient is nonzero, $w$ is not a scalar multiple of any single~$a_i$. By condition~(i), $\mathrm{wt}(w) \neq d$. Since $A_{d+1}(C) = 0$, we have $\mathrm{wt}(w) \geq d + 2$. As $v$ and $w$ differ in at most one coordinate, $\mathrm{wt}(v) \geq \mathrm{wt}(w) - 1 \geq d + 1$.
\end{itemize}
In all cases $\mathrm{wt}(v) \geq d$, and the scalar multiples of $a_1$ still achieve weight exactly~$d$, so $d_{\min}(D) = d$.

\textbf{Claim 3:} $D$ has exactly $t - 1$ independent weight-$d$ codewords. From Cases~B and~C, $\mathrm{wt}(v) \geq d + 1$ whenever $\mu \neq 0$. Hence every weight-$d$ codeword of $D$ has $\mu = 0$, meaning $v = w \in C$ is a weight-$d$ codeword of $C$, and therefore a scalar multiple of some $a_i$ with $i \leq t - 1$.
\end{proof}

 A sufficient condition for~(i) is $d_2(C) \geq 2d$, where $d_2(C)$ denotes the second generalized Hamming weight of $C$. Indeed, if $d_2(C) \geq 2d$, then any two independent weight-$d$ codewords $a_i, a_j$ span a $2$-dimensional subcode with support of size at least $2d$, forcing $\mathrm{supp}(a_i) \cap \mathrm{supp}(a_j) = \emptyset$. With disjoint supports, any nontrivial linear combination of two or more weight-$d$ basis vectors has weight at least $2d > d$, so no such combination can have weight~$d$.

\subsection{Two-Position Insertion}
\label{sec:reverse-two}

By strengthening the conditions on the weight distribution, we can additionally ensure that the constructed code has no codewords of weight $d + 1$. This is significant for the FCC application, as it results in a larger gap between the data protection distance and the function protection distance.

\begin{theorem}
\label{thm:reverse2}
Let $C$ be an $[n, k, d]_q$ code with $n,d \geq 3$ and $t > 1$ independent minimum-weight codewords. Suppose:
\begin{enumerate}
\item[\textup{(i)}] $A_d(C) = t(q - 1)$.
\item[\textup{(ii)}] $A_{d+1}(C) = A_{d+2}(C) = A_{d+3}(C) = 0$.
\end{enumerate}
Then there exists an $[n, k, d]_q$ code $D$ with exactly $t - 1$ independent minimum-weight codewords and $A_{d+1}(D) = 0$. In particular, $\dim \langle S_{d+1}(D) \rangle = t - 1$.
\end{theorem}

\begin{proof}
Condition~(i) ensures that every weight-$d$ codeword of $C$ is a scalar multiple of one of the $t$ independent weight-$d$ basis vectors $a_1, \ldots, a_t$.

 Choose two positions $p_1, p_2 \notin \mathrm{supp}(a_t)$, which exist since $\mathrm{wt}(a_t) = d$ and $n \geq d + 2$ for any code with $k, d \geq 3$, and nonzero scalars $\alpha_1, \alpha_2 \in \mathbb{F}_q^*$. Define $b = a_t + \alpha_1 \mathbf{e}_{p_1} + \alpha_2 \mathbf{e}_{p_2}$. Then $\mathrm{wt}(b) = d + 2$. Define the code
\[
D = \langle a_1, \ldots, a_{t-1}, b, e_1, \ldots, e_{k-t} \rangle.
\]

Every codeword $v \in D$ can be written as $v = \sum_{i=1}^{t-1} \lambda_i a_i + \mu b + \sum_{j=1}^{k-t} \nu_j e_j$, with a corresponding codeword $w = \sum_{i=1}^{t-1} \lambda_i a_i + \mu a_t + \sum_{j=1}^{k-t} \nu_j e_j \in C$. Since $b$ and $a_t$ differ only at positions $p_1$ and $p_2$, we have $v_{p_\ell} = w_{p_\ell} + \mu \alpha_\ell$ for $\ell = 1, 2$, and $v_i = w_i$ for all $i \notin \{p_1, p_2\}$.

\textbf{Claim 1:} $\dim(D) = k$. If $b$ were linearly dependent on $\{a_1, \ldots, a_{t-1}, e_1, \ldots, e_{k-t}\}$, then $b \in C$, and $b - a_t = \alpha_1 \mathbf{e}_{p_1} + \alpha_2 \mathbf{e}_{p_2} \in C$ would be a codeword of weight~$2 < d$ (since $d \geq 3$), a contradiction.

\textbf{Claim 2:} $d_{\min}(D) = d$. Consider a nonzero codeword $v \in D$ with corresponding $w \in C$.
\begin{itemize}
\item \emph{Case A} ($\mu = 0$): $v = w \in C$, so $\mathrm{wt}(v) \geq d$.
\item \emph{Case B} ($\mu \neq 0$, all $\lambda_i = 0$, all $\nu_j = 0$): $v = \mu b$, so $\mathrm{wt}(v) = d + 2$.
\item \emph{Case C} ($\mu \neq 0$, at least one $\lambda_i$ or $\nu_j$ nonzero): The codeword $w \in C$ is nonzero and is not a scalar multiple of any single~$a_i$. By condition~(i), $\mathrm{wt}(w) \neq d$. Since $A_{d+1}(C) = A_{d+2}(C) = A_{d+3}(C) = 0$, we have $\mathrm{wt}(w) \geq d + 4$. As $v$ and $w$ differ in at most two coordinates, $\mathrm{wt}(v) \geq \mathrm{wt}(w) - 2 \geq d + 2$.
\end{itemize}
In all cases $\mathrm{wt}(v) \geq d$, and the scalar multiples of $a_1$ still achieve weight~$d$, so $d_{\min}(D) = d$.

\textbf{Claim 3:} $D$ has exactly $t - 1$ independent weight-$d$ codewords. From Cases~B and~C, $\mathrm{wt}(v) \geq d + 2$ whenever $\mu \neq 0$. Hence every weight-$d$ codeword of $D$ has $\mu = 0$ and is therefore a scalar multiple of some $a_i$ with $i \leq t - 1$.

\textbf{Claim 4:} $A_{d+1}(D) = 0$. We show that no codeword of $D$ has weight $d + 1$.
\begin{itemize}
\item \emph{Case A} ($\mu = 0$): $v = w \in C$ and $\mathrm{wt}(w) \neq d + 1$ since $A_{d+1}(C) = 0$.
\item \emph{Case B} ($\mu \neq 0$, all others zero): $\mathrm{wt}(v) = d + 2 \neq d + 1$.
\item \emph{Case C} ($\mu \neq 0$, at least one other nonzero): $\mathrm{wt}(w) \geq d + 4$ and $\mathrm{wt}(v) \geq d + 2$, so $\mathrm{wt}(v) \neq d + 1$.
\end{itemize}
Since $A_{d+1}(D) = 0$, the codewords of weight at most $d + 1$ in $D$ are exactly the codewords of weight $d$, namely the scalar multiples of $a_1, \ldots, a_{t-1}$. Therefore $\dim \langle S_{d+1}(D) \rangle = \dim \langle S_d(D) \rangle = t - 1$.
\end{proof}

\subsection{Application to Strict Function-Correcting Codes}
\label{sec:reverse-fcc}

We now interpret Theorems~\ref{thm:reverse1} and~\ref{thm:reverse2} in the context of strict FCCs, using the Cayley graph characterization from Section~\ref{sec:cayley}.

Let $C$ be an $[n, k, d]_q$ code satisfying the conditions of Theorem~\ref{thm:reverse1}, and let $D$ be the code produced by the one-position insertion. Since $\dim \langle S_d(D) \rangle = t - 1$, Corollary~\ref{cor:linear-components} implies that $G_d(D)$ has $q^{k-t+1}$ connected components, each of size $q^{t-1}$. By Corollary~\ref{cor:linear-existence}, the code $D$ can serve as a strict $(f : d, d+1)$-FCC for any function $f$ with $|\mathrm{Im}(f)| \leq q^{k-t+1}$ and each preimage size a multiple of $q^{t-1}$.

Similarly, let $D$ be the code produced by the two-position insertion of Theorem~\ref{thm:reverse2}. Since $\dim \langle S_{d+1}(D) \rangle = t - 1$, the graph $G_{d+1}(D)$ has $q^{k-t+1}$ connected components, each of size $q^{t-1}$. The code $D$ can therefore serve as a strict $(f : d, d+2)$-FCC, providing a larger gap between the data and function protection distances.

In the case $k = t$, i.e., when $C$ is generated by its minimum-weight codewords, the code $C$ itself cannot serve as a strict FCC for any nontrivial function (as noted in Section~\ref{sec:cayley}). However, the code $D$ produced by the construction has $\dim \langle S_d(D) \rangle = t - 1 = k - 1$, and can serve as a strict FCC with $|\mathrm{Im}(f)| \leq q$. The two theorems are summarized in Table~\ref{tab:reverse}.

\begin{table}[t]
\centering
\renewcommand{\arraystretch}{1.3}
\caption{Summary of the converse construction and its FCC interpretation.}
\label{tab:reverse}
\begin{tabular}{c|c}
\hline
  Code $C$ & Code $D$ \\
\hline
\multicolumn{2}{c}{\emph{Theorem~\ref{thm:reverse1}: One-position insertion}} \\
\hline
  $\dim \langle S_d \rangle = \dim \langle S_{d+1} \rangle = t$ & $\dim \langle S_d \rangle = t - 1$ \\[2mm]
 $(f : d,\, d+2)$-FCC & $(f : d,\, d+1)$-FCC \\
$|\mathrm{Im}(f)| \leq q^{k-t}$ & $|\mathrm{Im}(f)| \leq q^{k-t+1}$ \\[2mm]
\hline
\multicolumn{2}{c}{\emph{Theorem~\ref{thm:reverse2}: Two-position insertion}} \\
\hline
  $\dim \langle S_d \rangle = \cdots = \dim \langle S_{d+3} \rangle = t$ & $\dim \langle S_{d+1} \rangle = t - 1$ \\[2mm]
$(f : d,\, d+4)$-FCC & $(f : d,\, d+2)$-FCC \\
$|\mathrm{Im}(f)| \leq q^{k-t}$ & $|\mathrm{Im}(f)| \leq q^{k-t+1}$ \\
\hline
\end{tabular}
\end{table}

The table illustrates a trade-off: the code $D$ supports more function values than $C$ (by a factor of~$q$), but with a smaller gap between $d_d$ and $d_f$. In particular, when $k = t$, the code $C$ admits no strict FCC at all (since $q^{k-t} = 1$), while the code $D$ admits a strict FCC with up to $q$ function values.

\section{Chain Codes}
\label{sec:chain}

We now introduce a family of codes that satisfy the conditions of Theorems~\ref{thm:reverse1} and~\ref{thm:reverse2}, and therefore serve as input to the constructions developed in the previous section.

A \emph{chain code} is a linear code that admits a generator matrix with a specific overlap structure: each row has weight~$d$, and the supports of consecutive rows overlap in exactly~$s$ positions, forming a chain of interlocking blocks. We consider an open and a closed variant, denoted $C_{\mathrm{o}}(k,d,s)_q$ and $C_{\mathrm{c}}(k,d,s)_q$ respectively, and defined in the following subsections.

\subsection{Open Chain Code}
\label{sec:chain-open}

\begin{definition}[Open Chain Code]
\label{def:chain-open}
Let $q$ be a prime power, $k \geq 2$, $d \geq 2$, and $0 \leq s < d$. The open chain code $C_{\mathrm{o}}(k, d, s)_q$ is the $[n, k, d]_q$ code with $n = kd - (k-1)s$, generated by $a_0, a_1, \ldots, a_{k-1} \in \mathbb{F}_q^n$ defined as
\[
(a_i)_j =
\begin{cases}
1 & \text{if } i(d - s) \leq j < i(d - s) + d, \\
0 & \text{otherwise},
\end{cases}
\qquad j \in \{0,1,\ldots, n-1\}.
\]
Each $a_i$ has weight $d$ and occupies positions $[i(d-s),\, i(d-s) + d)$. Consecutive rows $a_i$ and $a_{i+1}$ share exactly $s$ positions, while non-adjacent rows $a_i$ and $a_j$ with $|i - j| \geq 2$ have disjoint supports when $s < d/2$.
\end{definition}

\begin{example}
\label{ex:chain-open}
The open chain code $C_{\mathrm{o}}(3, 6, 1)_3$ has $n = 16$ and basis vectors
\begin{align*}
a_0 &= (1\,1\,1\,1\,1\,1\,0\,0\,0\,0\,0\,0\,0\,0\,0\,0), \\
a_1 &= (0\,0\,0\,0\,0\,1\,1\,1\,1\,1\,1\,0\,0\,0\,0\,0), \\
a_2 &= (0\,0\,0\,0\,0\,0\,0\,0\,0\,0\,1\,1\,1\,1\,1\,1).
\end{align*}
Here $a_0$ and $a_1$ overlap in position $\{5\}$, and $a_1$ and $a_2$ overlap in position $\{10\}$.
\end{example}

\subsection{Closed Chain Code}
\label{sec:chain-closed}

\begin{definition}[Closed Chain Code]
\label{def:chain-closed}
Let $q$ be a prime power, $k \geq 2$, $d \geq 2$, and $0 \leq s < d$. The closed chain code $C_{\mathrm{c}}(k, d, s)_q$ is the $[n, k, d]_q$ code with $n = k(d - s)$, generated by $a_0, a_1, \ldots, a_{k-1} \in \mathbb{F}_q^n$ defined as
\[
(a_i)_j =
\begin{cases}
1 & \text{if } j \bmod n \in \{i(d-s),\, i(d-s)+1,\, \ldots,\, i(d-s)+d-1\}, \\
0 & \text{otherwise},
\end{cases}
\qquad j \in \{0,1,\ldots, n-1\}.
\]
where the positions are taken modulo $n$. Each consecutive pair $(a_i, a_{i+1})$, including $(a_{k-1}, a_0)$, overlaps in exactly $s$ positions.
\end{definition}

The closed chain code is shorter than the open one by $s$ positions: $n_{\mathrm{c}} = k(d - s) = (kd - (k-1)s) - s = n_{\mathrm{o}} - s$. The additional overlap between $a_{k-1}$ and $a_0$ accounts for this saving.

\begin{example}
\label{ex:chain-closed}
The closed chain code $C_{\mathrm{c}}(3, 6, 1)_3$ has $n = 15$ and basis vectors
\begin{align*}
a_0 &= (1\,1\,1\,1\,1\,1\,0\,0\,0\,0\,0\,0\,0\,0\,0), \\
a_1 &= (0\,0\,0\,0\,0\,1\,1\,1\,1\,1\,1\,0\,0\,0\,0), \\
a_2 &= (1\,0\,0\,0\,0\,0\,0\,0\,0\,0\,1\,1\,1\,1\,1).
\end{align*}
Here $a_2$ occupies positions $\{10, 11, 12, 13, 14, 0\}$ (wrapping around), overlapping with $a_1$ in $\{10\}$ and with $a_0$ in $\{0\}$.
\end{example}

\subsection{Properties of Chain Codes}
\label{sec:chain-properties}

\begin{proposition}
\label{prop:chain1}
Let $C$ be either $C_{\mathrm{o}}(k, d, s)_q$ or $C_{\mathrm{c}}(k, d, s)_q$. Then
\begin{enumerate}
\item[\textup{(a)}] $C$ is an $[n, k, d]_q$ code with $k = t$, i.e., $C$ is generated by its minimum-weight codewords,
\item[\textup{(b)}] $A_d(C) = k(q - 1)$,
\item[\textup{(c)}] $A_{d+1}(C) = 0$,
\end{enumerate}
under the following conditions on the overlap~$s$:
\begin{itemize}
\item For the open chain code: $s \leq \lfloor (d-2)/2 \rfloor$.
\item For the closed chain code: $s \leq \min\!\left( \ \left\lfloor \dfrac{d-2}{2} \right \rfloor,\, \left\lfloor \dfrac{(k-1)d - 2}{2k} \right\rfloor \ \right)$.
\end{itemize}
\end{proposition}

\begin{proof}
In both the open and closed case, each basis vector $a_i$ has weight $d$. Adjacent pairs $(a_i, a_{i+1})$ (indices modulo $k$ in the closed case) overlap in exactly $s$ positions. Non-adjacent pairs $a_i, a_j$ with $\min(|i-j|, k - |i-j|) \geq 2$ (in the closed case) or $|i-j| \geq 2$ (in the open case) have disjoint supports, since the gap between consecutive support blocks is $d - 2s > 0$ (using $s < d/2$).

\textbf{Weight of nontrivial combinations.} Let $v = \sum_{i=0}^{k-1} c_i a_i$ with at least two nonzero coefficients. We partition the set of indices with $c_i \neq 0$ into maximal blocks of consecutive indices (cyclically in the closed case). Consider one such block $B = \{i, i+1, \ldots, i+\ell\}$ of length $|B| = \ell + 1$, and let $v_B = \sum_{j \in B} c_j a_j$.

If $|B|=\ell+1 \leq k$, cancellation can occur only at the $\ell$ overlap regions within the block, each of size $s$. The two endpoints of the block each contribute at least $d - s$ exclusive positions, and each interior vector contributes at least $d - 2s$ exclusive positions. In the worst case, all overlap positions cancel, giving
\[
\mathrm{wt}(v_B) \geq 2(d - s) + (\ell - 1)(d - 2s).
\]

For $v$ with $r \geq 2$ nonzero coefficients spread across $m \geq 1$ blocks $B_1, \ldots, B_m$, since non-adjacent vectors have disjoint supports, different blocks contribute to disjoint positions, and $\mathrm{wt}(v) = \sum_{j=1}^m \mathrm{wt}(v_{B_j})$.

\textbf{(a):} Every block with $|B_j| \geq 2$ contributes $\mathrm{wt}(v_{B_j}) \geq 2(d - s) > d$. If every block is a singleton, then $m = r \geq 2$ and $\mathrm{wt}(v) \geq 2d > d$. In all cases $\mathrm{wt}(v) > d$, so the minimum-weight codewords are exactly the scalar multiples of individual $a_i$, and $d_{\min}(C) = d$.

\textbf{(b):} Since the only weight-$d$ codewords are the scalar multiples $\lambda a_i$ for $\lambda \in \mathbb{F}_q^*$ and $0 \leq i \leq k-1$, we have $A_d(C) = k(q-1)$.

\textbf{(c):} We show that every codeword with at least two nonzero coefficients has weight at least $d + 2$.

For the open chain code, the minimum weight of such a codeword is $2(d - s)$, achieved by an adjacent pair with full cancellation in the overlap. The condition $s \leq \lfloor (d-2)/2 \rfloor$ gives $2(d - s) \geq 2(d - (d-2)/2) = d + 2$, so $A_{d+1}(C) = 0$.

For the closed chain code, two types of blocks must be considered. A non-full block of length $\ell + 1$ with $2 \leq \ell + 1 < k$ has two endpoints and behaves as in the open case, contributing weight at least $2(d - s)$. The condition $s \leq \lfloor (d-2)/2 \rfloor$ ensures $2(d-s) \geq d + 2$. A full block of length $k$, having all non-zero coefficients, forms a single cyclic block with no endpoints. Each vector contributes at least $d - 2s$ non-overlap positions, giving weight at least $k(d - 2s)$. The condition $s \leq \lfloor ((k-1)d - 2)/(2k) \rfloor$ ensures $k(d - 2s) \geq d + 2$. Taking the minimum of the two overlap bounds guarantees $A_{d+1}(C) = 0$ in all cases.
\end{proof}

All the results above remain valid if the all-ones entries in each $a_i$ are replaced by arbitrary nonzero elements of $\mathbb{F}_q$, since the weight bounds depend only on the support structure. Exclusive positions always contribute nonzero entries regardless of the specific values, and the worst-case cancellation count at overlap positions is unchanged.

Under a tighter overlap condition, the same codes also satisfy the hypotheses of Theorem~\ref{thm:reverse2}.

\begin{proposition}
\label{prop:chain2}
Let $d \geq 4$, and let $C$ be either $C_{\mathrm{o}}(k, d, s)_q$ or $C_{\mathrm{c}}(k, d, s)_q$. Then
\begin{enumerate}
\item[\textup{(a)}] $C$ is an $[n, k, d]_q$ code with $k = t$,
\item[\textup{(b)}] $A_d(C) = k(q - 1)$,
\item[\textup{(c)}] $A_{d+1}(C) = A_{d+2}(C) = A_{d+3}(C) = 0$,
\end{enumerate}
 under the following conditions on the overlap~$s$:
\begin{itemize}
\item For the open chain code: $s \leq \lfloor (d-4)/2 \rfloor$.
\item For the closed chain code: $s \leq \min\left(\ \left\lfloor \dfrac{d-4}{2} \right\rfloor,\, \left\lfloor \dfrac{(k-1)d - 4}{2k} \right\rfloor\ \right)$.
\end{itemize}
\end{proposition}

\begin{proof}
Parts~(a) and~(b) follow exactly as in the proof of Proposition~\ref{prop:chain1}. For part~(c), the same block-weight analysis applies with the target weight $d + 2$ replaced by $d + 4$: requiring $2(d - s) \geq d + 4$ yields $s \leq \lfloor (d-4)/2 \rfloor$, and requiring $k(d - 2s) \geq d + 4$ yields $s \leq \lfloor ((k-1)d - 4)/(2k) \rfloor$.
\end{proof}

\subsection{Parameters and Examples}
\label{sec:chain-params}

The parameters of chain codes are summarized in Table~\ref{tab:chain-params}. 

\begin{table}[t]
\centering
\renewcommand{\arraystretch}{1.3}
\caption{Parameters of chain codes.}
\label{tab:chain-params}
\begin{tabular}{c|c|c}
\hline
 & Open & Closed \\
\hline
Length $n$ & $kd - (k-1)s$ & $k(d - s)$ \\
Dimension & $k$ & $k$ \\
Min.\ distance & $d$ & $d$ \\
Rate $R = k/n$ & $\dfrac{k}{kd - (k-1)s}$ & $\dfrac{1}{d - s}$ \\[2mm]
\hline
\end{tabular}
\end{table}

\begin{example}
\label{ex:chain-instances}
Table~\ref{tab:chain-examples} lists some specific chain code instances.

\begin{table}[t]
\centering
\renewcommand{\arraystretch}{1.15}
\caption{Examples of chain codes.}
\label{tab:chain-examples}
\begin{tabular}{cccc}
\hline
Code & $[n, k, d]_q$ & $A_d$ & $A_{d+1}$ \\
\hline
$C_{\mathrm{o}}(2, 6, 1)_5$ & $[11, 2, 6]_5$ & $8$ & $0$ \\
$C_{\mathrm{c}}(2, 6, 1)_5$ & $[10, 2, 6]_5$ & $8$ & $0$ \\
$C_{\mathrm{o}}(3, 6, 1)_3$ & $[16, 3, 6]_3$ & $6$ & $0$ \\
$C_{\mathrm{c}}(3, 6, 1)_3$ & $[15, 3, 6]_3$ & $6$ & $0$ \\
$C_{\mathrm{o}}(4, 5, 1)_7$ & $[17, 4, 5]_7$ & $24$ & $0$ \\
$C_{\mathrm{c}}(4, 5, 1)_7$ & $[16, 4, 5]_7$ & $24$ & $0$ \\
\hline
\end{tabular}
\end{table}
\end{example}

Since the closed chain code has a stricter overlap condition than the open one, it is natural to ask which variant achieves shorter length for given $k$ and $d$. Table~\ref{tab:chain-compare} compares the two at their respective maximum allowable overlaps for $q = 5$ and $d = 10$.

\begin{table}[t]
\centering
\renewcommand{\arraystretch}{1.15}
\caption{Length comparison of open and closed chain codes at maximum overlap ($q = 5$, $d = 10$).}
\label{tab:chain-compare}
\begin{tabular}{c|cc|cc|c}
\hline
$k$ & $s_{\mathrm{o}}$ & $n_{\mathrm{o}}$ & $s_{\mathrm{c}}$ & $n_{\mathrm{c}}$ & $n_{\mathrm{o}} - n_{\mathrm{c}}$ \\
\hline
$2$  & $4$ & $16$ & $2$ & $16$ & $0$ \\
$3$  & $4$ & $22$ & $3$ & $21$ & $1$ \\
$4$  & $4$ & $28$ & $3$ & $28$ & $0$ \\
$5$  & $4$ & $34$ & $3$ & $35$ & $-1$ \\
$10$ & $4$ & $64$ & $4$ & $60$ & $4$ \\
\hline
\end{tabular}
\end{table}

The comparison shows that neither variant is uniformly shorter: for $k = 3$ the closed code saves one position, for $k = 5$ the open code is shorter, and for $k = 10$ the closed code saves four positions. The relative advantage depends on the interplay between the stricter overlap bound for the closed code and the length saving from the wraparound structure.

By Propositions~\ref{prop:chain1} and~\ref{prop:chain2}, every chain code satisfying the respective overlap conditions can be used as input to Theorems~\ref{thm:reverse1} and~\ref{thm:reverse2}. The resulting code $D$ has the same parameters $[n, k, d]_q$ but with $k - 1$ independent minimum-weight codewords, and can serve as a strict FCC as described in Section~\ref{sec:reverse-fcc}.

\section{BCH Code Construction}
\label{sec:bch}

We now present an independent construction of strict function-correcting codes from narrow-sense BCH codes with designed distance three. By the Cayley graph characterization of Section~\ref{sec:cayley}, the connected components of the $\alpha$-distance graph of a linear code are cosets of the subcode generated by codewords of weight at most $\alpha$. For the BCH code $C_{1,2}$ with $\alpha = d_f - 1 = 3$, this subcode is $\langle S_3 \rangle$, the span of all weight-3 codewords. If $\langle S_3 \rangle$ is contained in a proper subcode $D$ of $C_{1,2}$, then the number of connected components is at least $|C_{1,2}|/|D|$, which equals $p^{\mathrm{codim}(D)}$. The smaller $D$ is, the more components there are, and hence the more function values can be supported in the FCC construction.

Mogilnykh and Solov'eva~\cite{MS2020} proved that for primes $p \geq 5$, the weight-3 codewords of $C_{1,2}$ are contained in the proper subcode $C_{1,2,\,p^2+1}$, obtained by adjoining a single additional cyclotomic coset to the defining set. The following theorem strengthens this by showing that the weight-3 codewords lie in a much smaller subcode, obtained by adjoining $(m-1)/2$ additional cyclotomic cosets.


\begin{theorem}
\label{thm:bch}
Let $p \geq 5$ be a prime and $m \geq 3$ an odd integer. Then the set of weight-$3$ codewords of $C_{1,2}$ is contained in the subcode
\[
D \;=\; C_{1,2,\; p+1,\; p^2+1,\; \ldots,\; p^{(m-1)/2}+1}.
\]
In particular, $\langle S_3 \rangle \subseteq D \subsetneq C_{1,2}$, where $\langle S_3 \rangle$ denotes the $\mathbb{F}_p$-span of all weight-$3$ codewords of $C_{1,2}$.
\end{theorem}

\begin{proof}
Let $c(x)$ be a weight-3 codeword of $C_{1,2}$. Without loss of generality, we may write
\[
c(x) = 1 + a\,x^i + b\,x^j, \qquad a, b \in \mathbb{F}_p^*, \quad 0 < i < j \leq p^m - 2.
\]
Let $\alpha$ be a primitive element of $\mathbb{F}_{p^m}$. Since $c(x) \in C_{1,2}$, the parity-check conditions give
\begin{align}
1 + a\,\alpha^i + b\,\alpha^j &= 0, \label{eq:pc1} \\
1 + a\,\alpha^{2i} + b\,\alpha^{2j} &= 0. \label{eq:pc2}
\end{align}
We show that $\alpha^{p^r+1}$ is a root of $c(x)$ for every $1 \leq r \leq m-1$, i.e.,
\begin{equation}
\label{eq:target}
1 + a\,\alpha^{(p^r+1)i} + b\,\alpha^{(p^r+1)j} = 0.
\end{equation}

From~\eqref{eq:pc1}, we get $\alpha^j = -(1 + a\,\alpha^i)/b$. Substituting into~\eqref{eq:pc2} and simplifying yields
\begin{equation}
\label{eq:quadratic}
a(a + b)\,\alpha^{2i} + 2a\,\alpha^i + (b + 1) = 0.
\end{equation}

If $a + b = 0$), then equation~\eqref{eq:quadratic} reduces to $2a\,\alpha^i + (b+1) = 0$, so $\alpha^i = -(b+1)/(2a) \in \mathbb{F}_p$.

If $a + b \neq 0$, then equation~\eqref{eq:quadratic} is a nondegenerate quadratic in $\alpha^i$ over $\mathbb{F}_p$, with discriminant $\Delta = 4a^2 - 4a(a+b)(b+1) \in \mathbb{F}_p$. Setting $D^2 = \Delta$, we have $(D^2)^{p-1} = 1$, so $(D^{p-1})^2 = 1$ and $D^{p-1} \in \{1, -1\}$. Then $D^{p^2 - 1} = (D^{p-1})^{p+1} = (\pm 1)^{p+1} = 1$ (since $p$ is odd), which gives $D \in \mathbb{F}_{p^2}$. The quadratic formula yields $\alpha^i = (-2a \pm D)\big/\big(2a(a+b)\big) \in \mathbb{F}_{p^2}$.

Since $m$ is odd, $\gcd(2, m) = 1$, and therefore $\mathbb{F}_{p^2} \cap \mathbb{F}_{p^m} = \mathbb{F}_p$. As $\alpha^i$ is an element of $\mathbb{F}_{p^m}$ lying in $\mathbb{F}_{p^2}$, we conclude $\alpha^i \in \mathbb{F}_p$.

In both cases $\alpha^i \in \mathbb{F}_p$, and the same argument gives $\alpha^j \in \mathbb{F}_p$.

Since $\alpha^i \in \mathbb{F}_p$, we have $(\alpha^i)^{p^r} = \alpha^i$ for every $r \geq 1$, and similarly for $\alpha^j$. Therefore, for any $1 \leq r \leq m - 1$,
\[
1 + a\,\alpha^{(p^r+1)i} + b\,\alpha^{(p^r+1)j}
= 1 + a\,(\alpha^i)^{p^r}\!\cdot\alpha^i + b\,(\alpha^j)^{p^r}\!\cdot\alpha^j
= 1 + a\,\alpha^{2i} + b\,\alpha^{2j} = 0,
\]
where the last equality follows from~\eqref{eq:pc2}. Hence $\alpha^{p^r+1}$ is a root of $c(x)$ for every $1 \leq r \leq m - 1$.

For $1 \leq r \leq m - 1$, the element $p^r + 1$ belongs to the cyclotomic coset $\mathrm{Cl}(p^r + 1)$. Since $(p^r + 1)\,p^{m-r} = p^m + p^{m-r} \equiv p^{m-r} + 1 \pmod{p^m - 1}$, we have $\mathrm{Cl}(p^r + 1) = \mathrm{Cl}(p^{m-r} + 1)$. Because $m$ is odd, the pairs $\{r,\, m - r\}$ for $r = 1, \ldots, m-1$ yield at most $(m-1)/2$ distinct cyclotomic cosets, and we may take $1 \leq r \leq (m-1)/2$ as representatives. Therefore
\[
c(x) \in C_{1,2,\; p+1,\; p^2+1,\; \ldots,\; p^{(m-1)/2}+1} = D. \qedhere
\]
\end{proof}

The containment $\langle S_3 \rangle \subseteq D$ is proper, since $p \geq 5$, the element $p + 1$ does not belong to $\mathrm{Cl}(1) = \{p^i : 0 \leq i \leq m-1\}$ or $\mathrm{Cl}(2) = \{2p^i : 0 \leq i \leq m-1\}$. Every element of $\mathrm{Cl}(1)$ is divisible by $p$, and every element of $\mathrm{Cl}(2)$ other than $2$ is divisible by $p$, but $p + 1$ is neither divisible by $p$ nor equal to $2$. Hence $D$ is a proper subcode of $C_{1,2}$.

When $m$ is prime, the dimension of $D$ can be determined exactly.

\begin{theorem}
\label{thm:bch-dim}
Let $p \geq 5$ be a prime and $m \geq 3$ a prime. Then
\[
\dim(D) \;=\; p^m - 1 - \frac{m(m+3)}{2}.
\]
\end{theorem}

\begin{proof}
Since $\dim(C_{1,2}) = p^m - 1 - 2m$, it suffices to show that the $(m-1)/2$ additional cyclotomic cosets $\mathrm{Cl}(p^r + 1)$ for $1 \leq r \leq (m-1)/2$ each have size $m$ and are pairwise distinct and disjoint from $\mathrm{Cl}(1)$ and $\mathrm{Cl}(2)$. We verify three claims.

\textit{Claim~1:} $|\mathrm{Cl}(p^r + 1)| = m$ for all $1 \leq r \leq (m-1)/2$.

Since $|\mathrm{Cl}(i)|$ divides $m$ and $m$ is prime, $|\mathrm{Cl}(i)| \in \{1, m\}$. If $|\mathrm{Cl}(p^r + 1)| = 1$, then $p(p^r + 1) \equiv p^r + 1 \pmod{p^m - 1}$, giving $(p - 1)(p^r + 1) \equiv 0 \pmod{p^m - 1}$. Since $r \leq (m-1)/2$ and $m \geq 3$, we have $(p-1)(p^r + 1) \leq (p-1)(p^{(m-1)/2} + 1) < p^m - 1$, a contradiction. Hence $|\mathrm{Cl}(p^r + 1)| = m$.

\textit{Claim~2:} $\mathrm{Cl}(p^r + 1) \neq \mathrm{Cl}(p^s + 1)$ for all $1 \leq r < s \leq (m-1)/2$.

Suppose $\mathrm{Cl}(p^r + 1) = \mathrm{Cl}(p^s + 1)$. Then $p^t(p^r + 1) \equiv p^s + 1 \pmod{p^m - 1}$ for some $t \geq 1$, giving $p^{r+t} + p^t \equiv p^s + 1 \pmod{p^m - 1}$. Since $t \neq 0$, this forces $r + t \equiv 0 \pmod{m}$ and $t \equiv s \pmod{m}$, hence $r + s \equiv 0 \pmod{m}$. But $2 \leq r + s \leq m - 1$, so this is impossible.

\textit{Claim~3:} $\mathrm{Cl}(p^r + 1) \neq \mathrm{Cl}(1)$ and $\mathrm{Cl}(p^r + 1) \neq \mathrm{Cl}(2)$ for all $1 \leq r \leq (m-1)/2$.

Every element of $\mathrm{Cl}(1)$ is divisible by $p$, and every element of $\mathrm{Cl}(2)$ except $2$ is divisible by $p$. Since $p^r + 1$ is not divisible by $p$ and is not equal to $2$ (as $p \geq 5$ and $r \geq 1$), the element $p^r + 1$ does not belong to $\mathrm{Cl}(1)$ or $\mathrm{Cl}(2)$.

The three claims together give
\[
\dim(D) = p^m - 1 - \Big(\underbrace{2m}_{\mathrm{Cl}(1),\,\mathrm{Cl}(2)} + \underbrace{\frac{m-1}{2} \cdot m}_{\mathrm{Cl}(p^r+1), r \in [\frac{m-1}{2}]}\Big) = p^m - 1 - \frac{m(m + 3)}{2}. 
\]
\end{proof}

\begin{corollary}
\label{cor:bch-dim}
Under the hypotheses of Theorem~\textup{\ref{thm:bch}}, $\dim\langle S_3 \rangle \leq p^m - 1 - m(m+3)/2$. When $m$ is prime, this bound equals $\dim(D)$.
\end{corollary}

The codimension of $D$ in $C_{1,2}$ is
\[
\dim(C_{1,2}) - \dim(D) = (p^m - 1 - 2m) - \Big(p^m - 1 - \frac{m(m+3)}{2}\Big) = \frac{m(m-1)}{2},
\]
so $C_{1,2}$ partitions into $p^{m(m-1)/2}$ cosets of $D$, each of size $|D| = p^{\dim(D)}$.

\subsection{Application to FCC Construction}
\label{sec:bch-fcc}

Since $\langle S_3 \rangle \subseteq D \subsetneq C_{1,2}$, the coset partition of $\langle S_3 \rangle$ in $C_{1,2}$ is a refinement of the coset partition of $D$ in $C_{1,2}$. In particular, the connected components of $G_3(C_{1,2})$ are unions of cosets of $D$. This makes $C_{1,2}$ suitable for use as a strict $(f : 3, 4)$-FCC. The data-correcting distance is $d_d = d_{\min}(C_{1,2}) = 3$, and the function-correcting distance is $d_f = 4$, since weight-4 codewords exist in $C_{1,2} \setminus D$.

More precisely, for any function $f \colon \mathbb{F}_p^k \to \mathrm{Im}(f)$ with $k = \dim(C_{1,2}) = p^m - 1 - 2m$ satisfying
\begin{enumerate}
\item[\textup{(i)}] $|\mathrm{Im}(f)| \leq p^{m(m-1)/2}$, and
\item[\textup{(ii)}] each $|f^{-1}(a)|$ is a multiple of $|D| = p^{\,p^m - 1 - m(m+3)/2}$,
\end{enumerate}
the code $C_{1,2}$ serves as a strict $(f : 3, 4)$-FCC.

\begin{example}
\label{ex:bch}
Let $p = 5$ and $m = 3$, so that $n = 124$ and $\mathbb{F}_{125} = \mathbb{F}_5[x]/(x^3 + x^2 + 1)$ with primitive element $\alpha$. The relevant cyclotomic cosets modulo $124$ are
\[
\mathrm{Cl}(1) = \{1, 5, 25\}, \qquad \mathrm{Cl}(2) = \{2, 10, 50\}, \qquad \mathrm{Cl}(6) = \{6, 30, 26\}.
\]
Here $p + 1 = 6$, and $(m-1)/2 = 1$, so the only additional coset is $\mathrm{Cl}(6)$. The two codes are
\begin{align*}
C_{1,2} &: \text{defining set } \mathrm{Cl}(1) \cup \mathrm{Cl}(2), \quad [124, 118, 3]_5, \\
D = C_{1,2,6} &: \text{defining set } \mathrm{Cl}(1) \cup \mathrm{Cl}(2) \cup \mathrm{Cl}(6), \quad [124, 115, 3]_5.
\end{align*}
The codimension is $118 - 115 = 3 = m(m-1)/2$, giving $5^3 = 125$ cosets of $D$ in $C_{1,2}$, each of size $5^{115}$.

Let $A$ be any $118 \times 3$ matrix of rank $3$ over $\mathbb{F}_5$, and define $f \colon \mathbb{F}_5^{118} \to \mathbb{F}_5^3$ by $f(u) = uA$. Then $|\mathrm{Im}(f)| = 125$ and $|f^{-1}(a)| = 5^{115}$ for every $a \in \mathrm{Im}(f)$. By Theorem~\ref{thm:bch}, the code $C_{1,2}$ is a strict $(f : 3, 4)$-FCC for this $f$, with parameters as follows:
\[
\renewcommand{\arraystretch}{1.15}
\begin{array}{c|c}
\hline
\text{Dimension } k & 118 \\
\text{Length } n & 124 \\
\text{Redundancy } r & 6 \\
\text{Function } f & \mathbb{F}_5^{118} \to \mathbb{F}_5^3 \\
|\mathrm{Im}(f)| & 125 \\
|f^{-1}(a)| & 5^{115} \\
(d_d,\, d_f) & (3, 4) \\
\hline
\end{array}
\]
\end{example}

\section{Conclusion}
\label{sec:conclusion}

We studied the existence and construction of strict function-correcting codes with data protection for linear codes. Using the Cayley graph structure of the $\alpha$-distance graph, we showed that the existence problem reduces to finding codes whose minimum-weight codewords generate a proper subcode. We developed a converse to Simonis's theorem that transforms a code generated by its minimum-weight codewords into one with the same parameters but fewer independent minimum-weight codewords, and introduced chain codes as an infinite family satisfying the required conditions. Independently, we proved that for primes $p \geq 5$ and odd $m \geq 3$, the weight-3 codewords of the narrow-sense BCH code $C_{1,2}$ are contained in a proper subcode of codimension $m(m-1)/2$, resulting in  explicit strict $(f:3,4)$-FCC constructions.


 

\end{document}